\documentclass[journal]{IEEEtran}

\usepackage{amsmath,amssymb,amsthm,mathtools}
\usepackage{booktabs}
\usepackage{graphicx}
\usepackage{xcolor}
\usepackage{soul}
\usepackage{cite}

\sethlcolor{yellow}
\newcommand{\HL}[1]{#1}
\soulregister\cite7
\soulregister\ref7
\soulregister\eqref7

\newcommand{\Prb}{\Pr}

\newcommand{\tr}{\operatorname{tr}}

\theoremstyle{plain}

\newtheorem{proposition}{Proposition}
\newtheorem{corollary}{Corollary}

\theoremstyle{remark}

\title{Selection-Induced Contraction of Innovation Statistics in Gated Kalman Filters}

\author{Barak Or,~\IEEEmembership{Member, IEEE}%
\thanks{B. Or is with metaor artificial intelligence, Haifa, Israel,
and with the Google Reichman Tech School, Reichman University, Herzliya, Israel.
E-mail: barakorr@gmail.com}}

\begin{document}
\maketitle

\begin{abstract}
Validation gating is a fundamental component of classical Kalman-based tracking
systems. Only measurements whose normalized innovation squared (NIS) falls below
a prescribed threshold are considered for state update. While this procedure is
statistically motivated by the chi-square distribution, it implicitly replaces
the unconditional innovation process with a conditionally observed one,
restricted to the validation event.

This paper shows that innovation statistics computed after gating converge to
gate-conditioned rather than unconditional nominal reference quantities.
Under classical linear-Gaussian assumptions, we derive exact expressions for the
first- and second-order moments of the innovation conditioned on ellipsoidal
gating, and show that gating induces a deterministic, dimension-dependent
contraction of the innovation covariance relative to the nominal reference.

The analysis is extended to nearest-neighbor (NN) association, which is shown to
act as an additional statistical selection operator. We prove that selecting the
minimum norm innovation among multiple in-gate measurements introduces an
unavoidable energy contraction, implying that nominal innovation reference
statistics cannot be preserved under nontrivial gating and association due to
deterministic selection effects, even under perfectly matched linear-Gaussian
Kalman filter assumptions.
\HL{Monte Carlo validation is provided for the analytic truncation and
order-statistic results, together with an end-to-end multi-time-step tracking
example showing that nominal NIS-based measurement-noise tuning learns a biased
measurement-noise scale whereas a gate-aware correction remains close to the
true scale.}
\end{abstract}

\begin{IEEEkeywords}
Kalman filtering, validation gating, normalized innovation squared (NIS),
nearest-neighbor association, target tracking, innovation statistics,
\HL{adaptive tuning, normalized estimation error squared (NEES)}.
\end{IEEEkeywords}

\section{Introduction}

Kalman-based tracking systems are a cornerstone of modern aerospace, radar, and
navigation applications, where reliable state estimation is required under
uncertainty.
A central role in such systems is played by innovation-based statistics, which
are used for measurement validation, data association, consistency monitoring,
and adaptive tuning.
Among these, the normalized innovation squared (NIS) is particularly attractive
due to its simple statistical characterization under nominal linear Gaussian
assumptions: in the absence of additional selection mechanisms, the NIS follows
a chi-square distribution whose moments provide natural reference values for
diagnostic and tuning procedures.

In operational tracking systems, however, innovation statistics are rarely
observed in their unconditional form.
Prior to data association and state update, measurements are typically subjected
to ellipsoidal validation gating, whereby only innovations whose Mahalanobis
distance falls below a prescribed threshold are accepted.
Validation gating is widely justified on statistical and practical grounds, as
it limits computational complexity and suppresses gross outliers.
As a result, nearly all innovation-based diagnostics used in practice operate on
a post-gate innovation stream rather than on the nominal innovation process
assumed in classical Kalman filter theory.

Despite its ubiquity, the statistical consequences of validation gating are
rarely treated explicitly.
Once gating is applied, the innovations that enter association, filtering, and
diagnostic logic are no longer samples from the nominal Gaussian distribution,
but from a distribution truncated to the validation region.
Consequently, innovation-based statistics computed after gating need not satisfy
the classical reference properties associated with the chi-square law.
In particular, empirical NIS statistics obtained from accepted measurements may
systematically deviate from their nominal expectations even when the underlying
system model and noise statistics are perfectly matched.

A closely related issue arises in data association.
When multiple measurements fall inside the validation gate, NN association is
commonly employed to select a single measurement for the state update.
While computationally efficient and widely used, NN association further
conditions the observed innovation through order-statistic selection, as it
selects the innovation with minimum normalized distance among multiple
candidates.
The combined effect of validation gating and NN association therefore induces a
multi-stage selection process whose impact on innovation statistics is not
captured by nominal Kalman filter assumptions.

The objective of this paper is to isolate and characterize the statistical
effects induced solely by validation gating and NN association,
independent of clutter, false alarms, nonlinear dynamics, or modeling errors.
Within the classical linear-Gaussian Kalman filtering framework, we provide an
exact characterization of the innovation moments conditioned on ellipsoidal
gating and show that validation gating induces a deterministic,
dimension-dependent contraction of the innovation covariance relative to the
nominal reference.
We further show that NN association acts as an additional statistical selection
operator and introduces an unavoidable, multiplicity-dependent contraction of
innovation energy through order-statistic selection.
Together, these results imply that nominal innovation statistics cannot be
preserved under nontrivial gating and association, even when all Kalman filter
assumptions are satisfied.

\HL{The analysis is complemented by a Monte Carlo study with a multi-time-step
constant-velocity tracking problem.}
\HL{This study compares ungated filtering, gated filtering with fixed covariance,
nominal NIS-based covariance tuning, and a gate-aware NIS correction, and reports
NIS, corrected NIS, NEES, position RMSE, update rate, and final measurement-noise
scale.}

\subsection{Related Work}

Innovation-based statistics are a fundamental component of Kalman filtering
theory and practice, where quantities such as the NIS are routinely used for
consistency monitoring, measurement validation, and adaptive tuning.
Under nominal linear Gaussian assumptions, the statistical properties of the
innovation process are well understood and documented in classical estimation
references \cite{kalman1960new,kalman1961new,jazwinski2007stochastic,bar2004estimation,brown1992introduction}.

In practical tracking systems, innovation statistics are commonly employed in
adaptive filtering and noise covariance estimation schemes, often relying on
assumed chi-square properties of the NIS
\cite{mehra1970identification,akhlaghi2017adaptive,wang2018variational,youn2020novel}.
Such approaches implicitly assume that the observed innovation sequence is
representative of the nominal innovation distribution.
However, in operational settings, innovation statistics are almost always
evaluated after validation gating and data association, conditions under which
these assumptions may no longer hold.

Validation gating and NN association are standard components of tracking
systems in aerospace, radar, and navigation applications
\cite{bar2004estimation,ramachandra2018kalman,shats1988discrete,karlsson2017future}.
Despite their widespread use, the statistical impact of these selection
mechanisms on innovation-based diagnostics has received comparatively little
explicit theoretical treatment.

While truncation of Gaussian distributions is classical, the implications of
mandatory truncation and order-statistic selection on innovation-based
consistency diagnostics have not been explicitly characterized in the Kalman
tracking literature.
In particular, existing analyses typically focus on detection performance or
association accuracy, rather than on the induced bias in innovation statistics
themselves.

\HL{A closely related recent direction is the use of matrix-valued tracking
statistics and Wishart-based consistency measures.}
Forsling \emph{et al.} \cite{forsling2025matrix}
\HL{emphasize that scalar consistency statistics compress vector-valued tracking
information, and propose matrix-valued measures for offline and online
assessment.}
\HL{The present paper is complementary: it studies how the commonly used scalar
NIS itself is altered by mandatory gating and association before any consistency
test or adaptive tuning rule sees the data.}

Recent work on adaptive and learning-enhanced Kalman filtering further
highlights the reliance on innovation statistics for online consistency
assessment and tuning \cite{or2021kalman,or2022hybrid,or2023learning}.
These methods benefit from a precise understanding of how structural elements of
the tracking pipeline, such as gating and association, affect the observed
innovation process.

The present work complements the existing literature by providing an explicit
statistical characterization of gate-conditioned innovation moments and by
showing that NN association introduces an unavoidable order-statistic bias.
Unlike prior adaptive or robust filtering approaches, the results here isolate
selection-induced effects under ideal modeling assumptions, thereby clarifying
fundamental limitations in the interpretation of innovation-based diagnostics.

The main contributions of this paper are summarized as follows:
\begin{itemize}
\item Validation gating is interpreted as a statistical conditioning operation,
and exact expressions are derived for the first- and second-order moments of the
gate-conditioned innovation under linear-Gaussian assumptions.
\item It is shown that ellipsoidal gating induces a deterministic contraction of
the innovation covariance that depends only on the gate threshold and the
measurement dimension.
\item NN association is characterized as an order-statistic selection mechanism,
and it is proven that this selection introduces an additional, unavoidable
contraction of innovation energy, leading to an impossibility result: nominal
innovation statistics cannot be preserved under nontrivial gating and
association.
\item \HL{Monte Carlo simulations validate the analytic
gate-conditioned and NN order-statistic expressions, including confidence
intervals and direct theory-versus-simulation comparisons.}
\item \HL{An end-to-end tracking study shows how post-gate NIS bias
propagates into adaptive measurement-noise tuning and affects RMSE and NEES.}
\end{itemize}

The remainder of the paper is organized as follows.
Section~II introduces the problem formulation and the innovation model.
Section~III analyzes validation gating as a statistical selection mechanism and
derives gate-conditioned innovation moments.
Section~IV discusses the implications of these results for innovation-based
consistency diagnostics.
Section~V analyzes the additional bias induced by NN association.
\HL{Section~VI presents closed-form two-dimensional results, Monte Carlo
validation, and an end-to-end tracking study.}
Section~VII concludes the paper and discusses implications and limitations.

\section{Problem Formulation}

\subsection{Kalman Filter Innovation Model}

The analysis is conducted within the classical linear Gaussian Kalman filtering
framework.
\HL{Throughout the paper, the nominal assumptions are the following: the dynamic
and measurement models used by the filter are correctly specified; the process
and measurement noises are zero-mean, white, Gaussian, mutually independent, and
independent of the initial estimation error; and the covariance matrices used by
the filter match the true covariances.}
\HL{These assumptions are stated explicitly to distinguish the structural
selection effects studied here from model mismatch, non-Gaussian disturbances,
or suboptimal covariance tuning.}

Consider the discrete-time linear measurement model
\begin{equation}
z_k = H x_k + v_k,
\end{equation}
where $x_k \in \mathbb{R}^n$ denotes the system state, $z_k \in \mathbb{R}^m$
the measurement vector, $H \in \mathbb{R}^{m \times n}$ the measurement matrix,
and
\HL{the measurement noise is zero-mean, white, and Gaussian with covariance}
$R$.

Let $\hat{x}_k^-$ and $P_k^-$ denote the predicted state estimate and error
covariance produced by a Kalman filter.
The innovation is defined as
\begin{equation}
\nu_k \triangleq z_k - H \hat{x}_k^-,
\label{eq:innovation_definition}
\end{equation}
with associated innovation covariance
\begin{equation}
S_k \triangleq H P_k^- H^\top + R.
\label{eq:innovation_covariance}
\end{equation}

Under the nominal assumptions stated above, the innovation sequence
$\{\nu_k\}$ is zero-mean, Gaussian, and white.
In the sequel, we analyze a single time index and omit the subscript $k$ for
notational clarity.
Accordingly, we model the innovation as the random vector
\begin{equation}
\nu : \Omega \rightarrow \mathbb{R}^m,
\qquad
\nu \sim \mathcal{N}(0,S),
\label{eq:innovation_distribution}
\end{equation}
where $S \in \mathbb{S}_{++}^m$ denotes the symmetric positive-definite
innovation covariance matrix.

\subsection{Post-Gate Innovation Distribution}

Under the nominal linear Gaussian Kalman filter assumptions, the normalized
innovation squared (NIS) is defined as
\begin{equation}
Z \triangleq \nu^\top S^{-1}\nu,
\end{equation}
and follows a chi-square distribution with $m$ degrees of freedom.

Let
\begin{equation}
\mathcal{A} \triangleq \{ Z \le \tau \}
\label{eq:gating_event}
\end{equation}
denote the validation event induced by ellipsoidal gating.
All innovation samples that pass the gate are therefore distributed according to
the conditional distribution of $\nu$ given $\mathcal{A}$, denoted by
$\nu \mid \mathcal{A}$.
This conditional distribution corresponds to a Gaussian law truncated to the
ellipsoidal region
\[
\mathcal{E}_\tau
\triangleq
\{ \nu \in \mathbb{R}^m : \nu^\top S^{-1}\nu \le \tau \}
\]
rather than the unconditional Gaussian distribution in
\eqref{eq:innovation_distribution}.
The statistical properties of this gate-conditioned innovation distribution
form the basis of the analysis in the remainder of this paper.

\section{Innovation Gating as a Statistical Selection Mechanism}

We interpret ellipsoidal gating as a statistical selection mechanism acting on
the innovation process.
Specifically, gating restricts the observed innovation stream to realizations
satisfying the acceptance event $\mathcal{A}$ defined in
\eqref{eq:gating_event}.
Consequently, all innovations that enter data association, state update, and
diagnostic logic are drawn from the conditional distribution
$\nu \mid \mathcal{A}$, rather than from the nominal Gaussian distribution
$\mathcal{N}(0,S)$.

\subsection{Distribution of the NIS and Ellipsoidal Gating}

\begin{proposition}[Distribution of the NIS]
\label{prop:NIS_chi_square}
If $\nu \sim \mathcal{N}(0,S)$ with $S \in \mathbb{S}_{++}^m$, then
\begin{equation}
Z \sim \chi^2_m.
\end{equation}
\end{proposition}

This classical result provides the probabilistic basis for ellipsoidal
validation gating: selecting the threshold $\tau$ as a chi-square quantile
ensures
\begin{equation}
\mathbb{P}\{\mathcal{A}\} = P_g
\end{equation}
under the nominal innovation model.

Once validation gating is applied, however, the innovation process is no longer
observed unconditionally, but only through realizations satisfying the
acceptance event $\mathcal{A} = \{ Z \le \tau \}$.
As a result, all innovation-based statistics computed after gating are
statistics of a conditionally observed random variable.
In the innovation space $\mathbb{R}^m$, the acceptance event $\mathcal{A}$
corresponds to the ellipsoidal region $\mathcal{E}_\tau$.
Validation gating therefore implements a deterministic truncation of the
innovation distribution to $\mathcal{E}_\tau$, retaining only realizations
within a fixed Mahalanobis radius.

\subsection{Gate-Conditioned Innovation Moments}

\begin{proposition}[Gate-Conditioned Innovation Moments]
\label{prop:gate_conditioned_moments}
Let $\nu \sim \mathcal{N}(0,S)$ with $S \in \mathbb{S}_{++}^m$, and let
$\mathcal{A} = \{ \nu^\top S^{-1}\nu \le \tau \}$ denote the validation event.
Then the gate-conditioned innovation satisfies
\begin{subequations}
\begin{align}
\mathbb{E}[\nu \mid \mathcal{A}] &= 0,
\label{eq:cond_mean_zero} \\
\mathbb{E}[\nu\nu^\top \mid \mathcal{A}] &= \gamma(\tau,m)\, S,
\label{eq:cond_covariance}
\end{align}
\end{subequations}
where the scalar contraction factor $\gamma(\tau,m)$ is given by
\begin{equation}
\gamma(\tau,m)
\triangleq
\frac{1}{m}\,
\mathbb{E}\!\left[ Z \mid Z \le \tau \right],
\qquad
Z \sim \chi^2_m,
\label{eq:gamma_definition}
\end{equation}
and satisfies $0 < \gamma(\tau,m) < 1$.
\end{proposition}

Proposition~\ref{prop:gate_conditioned_moments}
\HL{shows that ellipsoidal gating preserves the zero mean of the innovation
while deterministically shrinking its second-order moment by a scalar factor.}
\HL{The factor depends only on the gate threshold and the measurement dimension,
not on the particular innovation covariance matrix.}

\subsection{Gate-Conditioned NIS Statistics}

An immediate consequence of \eqref{eq:cond_covariance} is an explicit expression
for the mean NIS after gating.

\begin{corollary}[Gate-Conditioned Mean NIS]
\label{cor:nis_mean}
Under the assumptions of Proposition~\ref{prop:gate_conditioned_moments},
\begin{equation}
\mathbb{E}[Z \mid \mathcal{A}] = m\,\gamma(\tau,m),
\label{eq:cond_nis_mean}
\end{equation}
where $Z = \nu^\top S^{-1}\nu$.
\end{corollary}

In contrast to the nominal reference value $\mathbb{E}[Z]=m$, the expected NIS
computed from accepted measurements reflects the conditioning induced by the
validation gate.

\section{Gate-Aware Consistency and Interpretation of Innovation Statistics}

The results of Section~III establish that ellipsoidal gating alters the
statistical properties of the innovation process.
In this section, we examine the implications of these results for
innovation-based consistency diagnostics, with particular emphasis on NIS.

\subsection{Classical Use of NIS for Consistency Validation}

In Kalman-based tracking systems, the normalized innovation $Z_k$ is commonly
employed as a diagnostic quantity for filter consistency.
Under the nominal Kalman filter assumptions recalled in Section~II, the NIS
follows a chi-square distribution with $m$ degrees of freedom.
As a result, standard practice is to compare instantaneous NIS values to
chi-square confidence bounds, or to monitor the empirical mean of $\{Z_k\}$ over
time and compare it to the nominal reference value $m$, as commonly done in
online consistency testing \cite{piche2016online}.

If the empirical mean is significantly below $m$, the filter is often declared
conservative in the innovation space; if it is significantly above $m$, it is
declared inconsistent or underestimating innovation uncertainty.
\HL{This interpretation is valid for the unconditional innovation sequence, but
it is not valid without modification for the accepted post-gate sequence.}

\subsection{Post-Gate NIS Is Not Chi-Square}

As a direct consequence of the conditioning induced by ellipsoidal gating, the
normalized innovation squared no longer follows its nominal distribution.
Specifically,
\begin{equation}
Z \mid \mathcal{A} \;\not\sim\; \chi^2_m,
\end{equation}
but instead follows a truncated chi-square law obtained by renormalizing
$\chi^2_m$ on the interval $[0,\tau]$ \cite{dunik2020covariance}.

Accordingly, as established in
Proposition~\ref{prop:gate_conditioned_moments},
\begin{equation}
\mathbb{E}[Z \mid \mathcal{A}] = m\,\gamma(\tau,m),
\end{equation}
with $\gamma(\tau,m)\in(0,1)$.
The nominal reference value $\mathbb{E}[Z]=m$ therefore ceases to be valid
whenever gating is applied.

\subsection{Gate-Aware Interpretation of NIS Statistics}

The results above imply that standard NIS diagnostics must be interpreted
relative to gate-conditioned reference values.
Specifically, when NIS statistics are computed from accepted innovations only,
the correct reference mean is given by Corollary~\ref{cor:nis_mean}.
Equivalently, a gate-aware normalized statistic is defined as
\begin{equation}
Z_k^{\text{corr}} \triangleq \frac{1}{\gamma(\tau,m)}\,Z_k,
\label{eq:corrected_nis}
\end{equation}
which satisfies
\begin{equation}
\mathbb{E}[Z_k^{\text{corr}} \mid \mathcal{A}] = m.
\end{equation}

This normalization allows the empirical mean of accepted NIS samples to be
compared to the classical reference mean, while preserving the existing Kalman
filter recursion and gating logic.
\HL{The correction is a diagnostic and tuning-reference correction; it does not
change the innovation used in the state update.}

\subsection{Implications for Adaptive Tuning and Diagnostics}

Many adaptive noise-tuning and consistency-monitoring schemes rely, either
explicitly or implicitly, on innovation covariance estimates or NIS-based
statistics \cite{dunik2017noise}, and often enforce nominal chi-square
innovation behavior as a tuning objective \cite{chen2023kalman}.

If such schemes operate on post-gate data without accounting for the
conditioning induced by gating, they will systematically underestimate the
innovation covariance,
\HL{may reduce the estimated measurement-noise level, and may tighten the
effective validation logic by making valid measurements more likely to be
rejected.}
\HL{The resulting behavior is better described as an increased probability of
missed updates or rejection of valid measurements, rather than as an increase in
false detections.}
This feedback mechanism arises even under ideal modeling assumptions and is a
direct consequence of ignoring the gate-conditioned nature of the observed
innovation process.

The contraction factor $\gamma(\tau,m)$ admits intuitive limiting behavior.
As $\tau \to \infty$, validation gating becomes inactive and
$\gamma(\tau,m) \to 1$, recovering the nominal innovation statistics.
Conversely, as $\tau \to 0$, only vanishingly small innovations are accepted and
$\gamma(\tau,m) \to 0$, driving the post-gate innovation energy to zero.

\section{Bias Induced by Nearest-Neighbor Association}

The previous sections characterized the statistical effect of ellipsoidal gating
on innovation moments.
In practical tracking systems, however, gating is only the first stage of a
selection process.
When multiple measurements fall inside the validation region, a data association
rule is applied to select a single measurement for the state update.
The most widely used rule in classical tracking is NN association
\cite{reid2003algorithm,bar1990tracking}.

This section shows that NN association introduces an additional and unavoidable
statistical bias that compounds the gate-induced effects.
Unlike modeling errors or tuning artifacts, this bias arises solely from
order-statistic selection and persists even under ideal nominal assumptions.

\subsection{Nearest-Neighbor Association as Statistical Selection}

Consider a time step at which $M \ge 1$ measurements pass the validation gate
$\mathcal{A}$.
Let $\{\nu^{(i)}\}_{i=1}^M$ denote the corresponding innovation vectors, each
generated according to the same post-gate distribution $(\nu \mid \mathcal{A})$.
NN association selects the innovation with minimum NIS,
\begin{equation}
i^\star \triangleq
\arg\min_{1 \le i \le M} \|\nu^{(i)}\|_{S^{-1}}^2,
\label{eq:nn_rule}
\end{equation}
and uses $\nu^{(i^\star)}$ for the update step.

While \eqref{eq:nn_rule} is often motivated algorithmically as a simple and
efficient approximation to more complex association schemes, it constitutes a
nonlinear statistical selection operator acting directly on the innovation.

\subsection{Energy Contraction Under NN Association}

The central effect of NN association is most clearly exposed at the level of
second-order energy.
The following proposition formalizes this effect without invoking Gaussianity or
any specific parametric form.

\begin{proposition}[Post-Gate NN Energy Contraction]
\label{prop:nn_energy_contraction}
Let $\nu \in \mathbb{R}^d$ be an innovation vector and let $\mathcal{A}$ denote
the gating acceptance event.
Let $\{\nu^{(i)}\}_{i=1}^M$ be independent samples from the conditional
distribution $(\nu \mid \mathcal{A})$, and let
$i^\star = \arg\min_i \|\nu^{(i)}\|$.
Then
\begin{equation}
\mathbb{E}\!\left[\|\nu^{(i^\star)}\|^2 \mid \mathcal{A}\right]
\;\le\;
\mathbb{E}\!\left[\|\nu\|^2 \mid \mathcal{A}\right],
\label{eq:nn_energy_bound}
\end{equation}
with strict inequality for $M>1$ whenever the conditional distribution
$(\nu \mid \mathcal{A})$ is non-degenerate.
\end{proposition}

This bound shows that NN association induces a systematic contraction of
innovation energy beyond that caused by gating alone.
The effect is purely statistical and follows from order-statistic selection.

\begin{corollary}[Impossibility of Preserving Nominal Innovation Energy Under Selection]
\label{cor:impossibility_nominal_energy}
Let $\nu \sim \mathcal{N}(0,S)$ with $S \in \mathbb{S}_{++}^m$ and let
$\mathcal{A}=\{\nu^\top S^{-1}\nu \le \tau\}$ be a nontrivial ellipsoidal gate
with $\Prb(\mathcal{A})\in(0,1)$.
Assume that, whenever the gate admits $M\ge 1$ candidate measurements, the
association rule selects the minimum-norm innovation among $M$ independent
post-gate candidates, yielding $\nu^{(i^\star)}$.
Then, for any $M>1$ for which $(\nu\mid \mathcal{A})$ is non-degenerate,
\begin{equation}
\mathbb{E}\!\left[\|\nu^{(i^\star)}\|^2 \mid \mathcal{A}\right]
\;<\;
\tr(S).
\label{eq:impossibility_energy}
\end{equation}
\end{corollary}

\begin{proof}
Proposition~\ref{prop:nn_energy_contraction} gives, for $M>1$ and
non-degenerate $(\nu\mid\mathcal{A})$,
\[
\mathbb{E}\!\left[\|\nu^{(i^\star)}\|^2 \mid \mathcal{A}\right]
<
\mathbb{E}\!\left[\|\nu\|^2 \mid \mathcal{A}\right].
\]
Moreover, Proposition~\ref{prop:gate_conditioned_moments} implies
\[
\mathbb{E}\!\left[\|\nu\|^2 \mid \mathcal{A}\right]
=
\tr\!\left(\gamma(\tau,m)S\right)
=
\gamma(\tau,m)\tr(S),
\]
with $\gamma(\tau,m)\in(0,1)$ for any nontrivial gate.
Therefore,
\[
\mathbb{E}\!\left[\|\nu^{(i^\star)}\|^2 \mid \mathcal{A}\right]
<
\gamma(\tau,m)\tr(S)
<
\tr(S),
\]
which proves the claim.
\end{proof}

\subsection{Interpretation via Order Statistics}

The purpose of this subsection is to provide intuition for
Proposition~\ref{prop:nn_energy_contraction} by isolating the purely statistical
mechanism underlying NN association.
NN association selects the minimum-norm element from a finite set of independent
post-gate innovations.
As a result, the selected innovation corresponds to the first-order statistic of
the squared norms.

For any nonnegative, non-degenerate random variable $X$ and $M \ge 2$
independent copies $\{X^{(i)}\}$,
\begin{equation}
\mathbb{E}\!\left[\min_i X^{(i)}\right] < \mathbb{E}[X].
\end{equation}
Applying this property to $X = \|\nu\|^2 \mid \mathcal{A}$ yields the strict
inequality in \eqref{eq:nn_energy_bound}.
Notably, this argument relies only on elementary properties of order statistics
and does not invoke Gaussianity, dimensionality, or the specific shape of the
validation gate.

\subsection{Other Data Association Methods}

\HL{The analysis focuses on NN association because it yields a clean
order-statistic statement.}
\HL{Other association methods do not remove the selection issue, but they expose
it in different mathematical forms.}
\HL{In probabilistic data association (PDA), the update uses likelihood-weighted
mixtures over gated measurements; the observed update is therefore affected by
both truncation and likelihood weighting rather than by a hard minimum alone.}
\HL{In multiple-hypothesis tracking (MHT), multiple hypotheses are propagated,
but gating, likelihood scoring, pruning, and hypothesis management still impose
selection on the set of innovations that influence the maintained tracks.}
\HL{Deriving exact gate-conditioned reference statistics for PDA and MHT requires
explicit assumptions on clutter, detection probability, and hypothesis
management, and is left as a natural extension.}

\section{Monte Carlo Validation and Tracking Case Study}

\HL{This section provides explicit illustrations, Monte Carlo validation, and a
multi-time-step tracking example.}
\HL{All simulations use white Gaussian process and measurement noises, correctly
specified covariances, and the same ellipsoidal validation rule used in the
analysis.}

\subsection{Two-Dimensional Closed-Form Gate Contraction}

For $m=2$, the NIS $Z = \nu^\top S^{-1}\nu$ follows a chi-square distribution
with two degrees of freedom under nominal linear Gaussian assumptions.
Equivalently, $Z$ is exponentially distributed with density
\begin{equation}
f_Z(z) = \frac{1}{2} e^{-z/2}, \qquad z \ge 0,
\end{equation}
and cumulative distribution function
\begin{equation}
F_Z(z) = 1 - e^{-z/2}.
\end{equation}

Given a validation gate with acceptance probability $P_g$, the gate threshold is
\begin{equation}
\tau = -2\ln(1-P_g).
\end{equation}
The gate-conditioned mean NIS is
\begin{equation}
\mathbb{E}[Z \mid Z \le \tau]
=
\frac{2 - (\tau+2)e^{-\tau/2}}{1-e^{-\tau/2}},
\end{equation}
or equivalently,
\begin{equation}
\mathbb{E}[Z \mid Z \le \tau]
=
2\left(1 + \frac{(1-P_g)\ln(1-P_g)}{P_g}\right).
\end{equation}
Thus
\begin{equation}
\gamma(P_g,2)
=
1 + \frac{(1-P_g)\ln(1-P_g)}{P_g}.
\end{equation}

\begin{table}[t]
\centering
\caption{Gate-induced contraction in two dimensions. Monte Carlo values are
reported with 95\% confidence half-widths.}
\label{tab:gate_validation}
\begin{tabular}{ccccc}
\toprule
$P_g$ & $\tau$ & $\gamma$ & Theory & MC \\
\midrule
0.90 & 4.605 & 0.744 & 1.488 & $1.487 \pm 0.004$ \\
0.95 & 5.991 & 0.842 & 1.685 & $1.684 \pm 0.005$ \\
0.99 & 9.210 & 0.953 & 1.907 & $1.907 \pm 0.006$ \\
\bottomrule
\end{tabular}
\end{table}

\begin{figure}[t]
\centering
\includegraphics[width=0.95\linewidth]{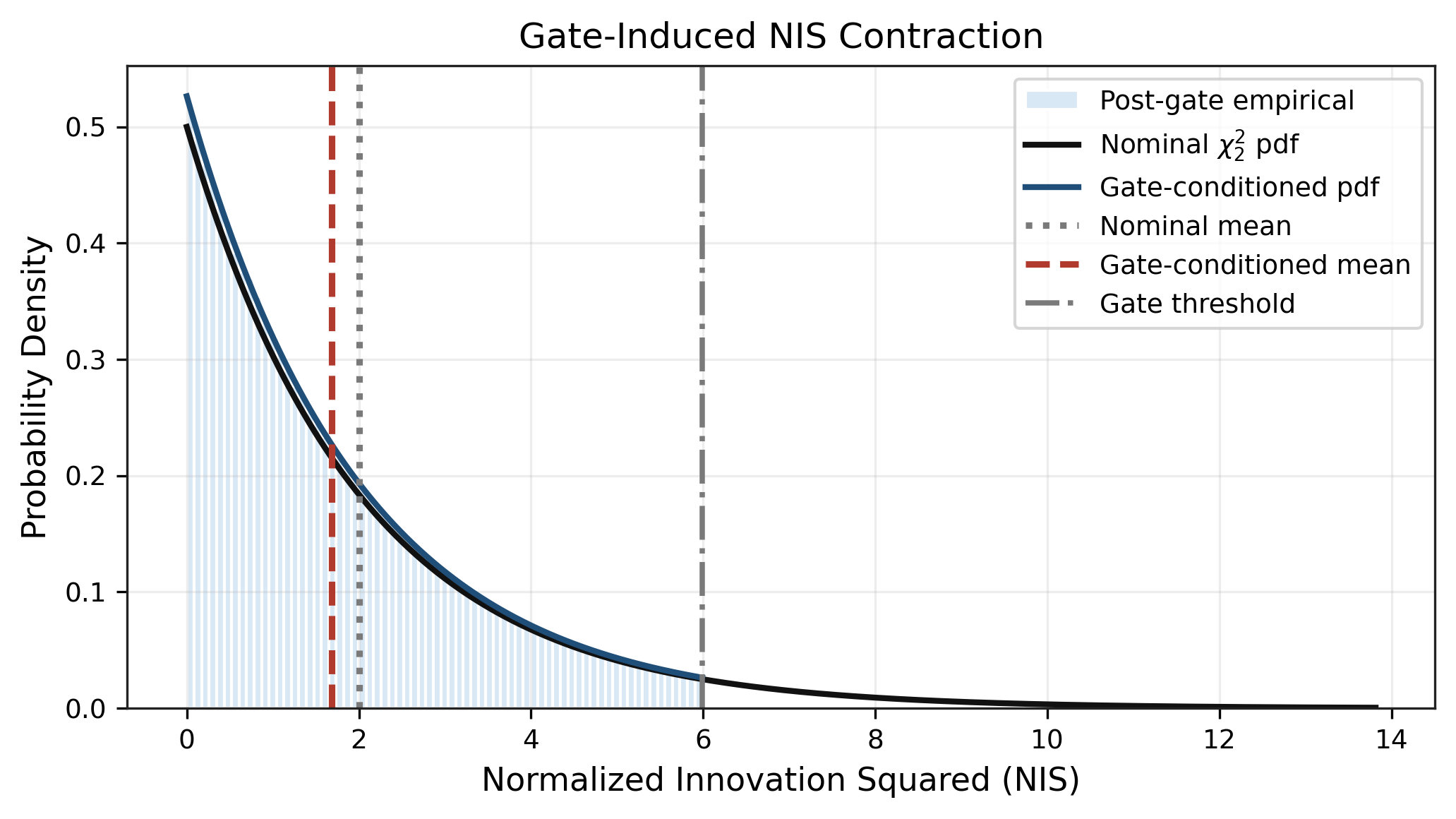}
\caption{\protect\HL{Monte Carlo validation of the gate-conditioned NIS distribution
for the two-dimensional case with gate probability 0.95. The accepted NIS
samples follow the truncated chi-square density, and their mean is lower than
the nominal reference.}}
\label{fig:gate_contraction_validation}
\end{figure}

Table~\ref{tab:gate_validation} and Fig.~\ref{fig:gate_contraction_validation}
\HL{show close agreement between the analytic gate-conditioned mean and Monte
Carlo estimates.}
\HL{For the common 0.95 gate probability, the expected accepted NIS is
approximately 1.685 rather than the nominal value 2.}

\subsection{NN Order-Statistic Validation}

\HL{For the NN experiment, independent samples are drawn from the
gate-conditioned NIS distribution and the minimum is selected among a fixed
number of in-gate candidates.}
For $Z_i \sim (Z\mid Z\le\tau)$ independent,
\begin{equation}
\mathbb{E}\!\left[\min_{1\le i\le M} Z_i\right]
=
\int_0^\tau
\left(
\frac{F_Z(\tau)-F_Z(z)}{F_Z(\tau)}
\right)^M dz.
\label{eq:min_order_mean}
\end{equation}

\begin{table}[t]
\centering
\caption{NN order-statistic contraction for $m=2$ and $P_g=0.95$.}
\label{tab:nn_validation}
\begin{tabular}{cccc}
\toprule
$M$ & Theory & MC & MC / nominal \\
\midrule
1  & 1.685 & $1.689 \pm 0.005$ & 0.845 \\
2  & 0.911 & $0.912 \pm 0.003$ & 0.456 \\
3  & 0.619 & $0.616 \pm 0.002$ & 0.308 \\
5  & 0.375 & $0.375 \pm 0.001$ & 0.187 \\
10 & 0.189 & $0.189 \pm 0.001$ & 0.094 \\
\bottomrule
\end{tabular}
\end{table}

\begin{figure}[t]
\centering
\includegraphics[width=0.95\linewidth]{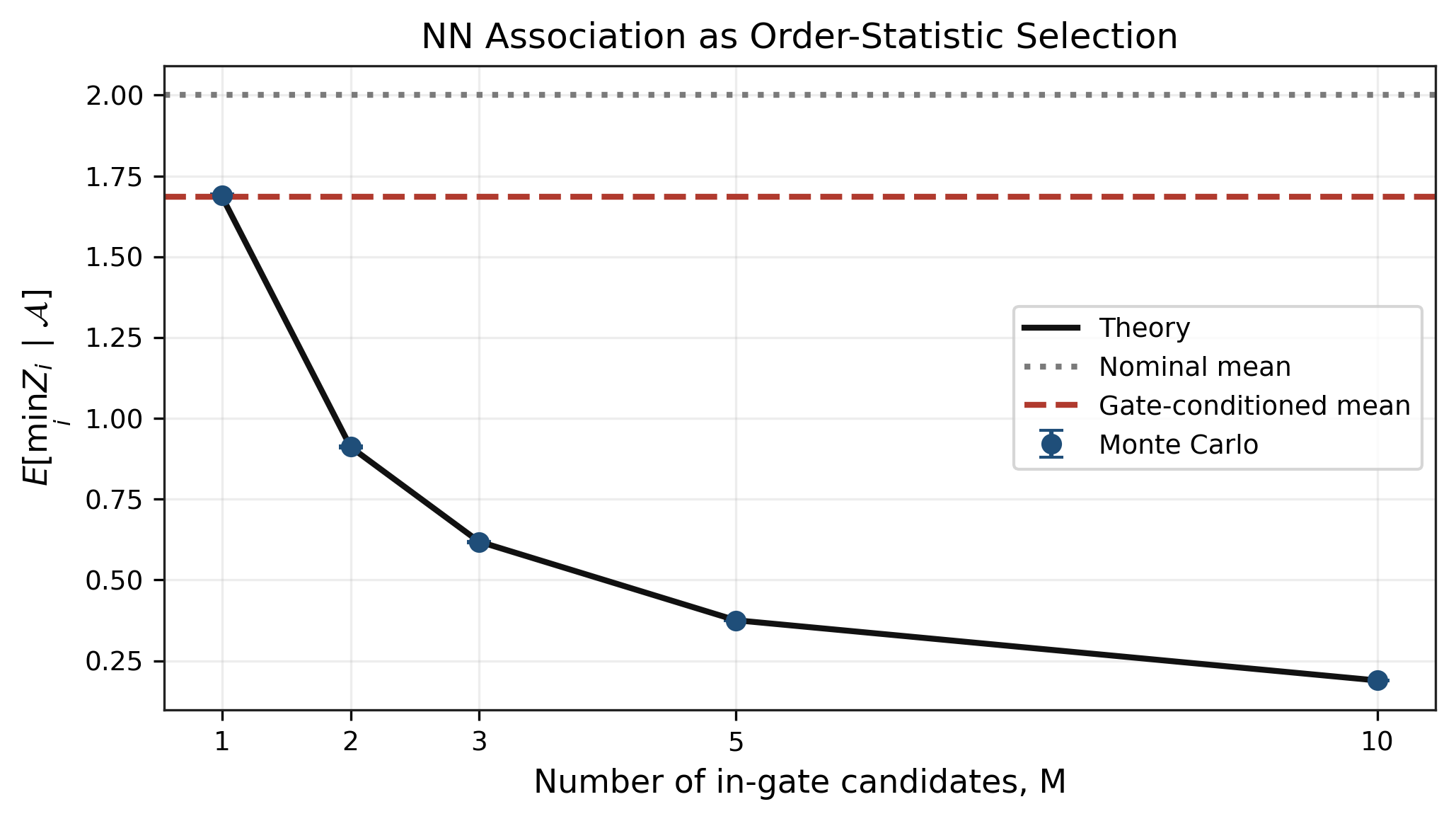}
\caption{\protect\HL{Expected selected NIS under NN association as a function of the
number of in-gate candidates. Monte Carlo results match the order-statistic
prediction and show rapid contraction as the number of candidates increases.}}
\label{fig:nn_order_statistic_validation}
\end{figure}

The results in Table~\ref{tab:nn_validation} and
Fig.~\ref{fig:nn_order_statistic_validation}
\HL{validate the order-statistic analysis.}
\HL{Even with two in-gate candidates, the expected selected NIS is reduced to
approximately 0.912, less than half of the nominal two-dimensional reference.}

\subsection{End-to-End Tracking and Tuning Study}

\HL{To illustrate the practical implication for adaptive tuning, a
multi-time-step constant-velocity tracking example is considered.}
The state is
\[
x_k =
\begin{bmatrix}
p_{x,k} & v_{x,k} & p_{y,k} & v_{y,k}
\end{bmatrix}^\top,
\]
with sampling period $\Delta t=1$.
The process and measurement models are
\begin{align}
x_k &= F x_{k-1} + w_{k-1},\\
z_k &= H x_k + v_k,
\end{align}
where $w_k$ and $v_k$ are zero-mean white Gaussian noises with correctly
specified covariances.
\HL{The measurement dimension is two, the gate probability is 0.95, and the
Monte Carlo study uses 1200 trials of 120 time steps.}
\HL{The compared cases are: no gate with fixed measurement covariance, gate with
fixed measurement covariance, gate with nominal NIS-based covariance tuning, and
gate with gate-aware NIS-based covariance tuning.}

\begin{table*}[t]
\centering
\caption{End-to-end tracking results. Values are Monte Carlo means with 95\%
confidence half-widths.}
\label{tab:tracking_results}
\begin{tabular}{lcccccc}
\toprule
Case & Used NIS & Corrected NIS & Final $R$ scale & Update rate & RMSE & NEES \\
\midrule
No gate, fixed $R$
& $1.986\pm0.011$ & $1.986\pm0.011$ & $1.000\pm0.000$
& $1.000\pm0.000$ & $3.284\pm0.019$ & $3.970\pm0.036$ \\
Gate, fixed $R$
& $1.717\pm0.008$ & $2.039\pm0.010$ & $1.000\pm0.000$
& $0.936\pm0.004$ & $5.444\pm1.138$ & $4.477\pm0.073$ \\
Gate, nominal NIS tuning
& $1.838\pm0.006$ & $2.182\pm0.007$ & $0.797\pm0.006$
& $0.910\pm0.004$ & $4.890\pm0.776$ & $4.962\pm0.076$ \\
Gate, gate-aware NIS tuning
& $1.706\pm0.005$ & $2.025\pm0.006$ & $1.042\pm0.008$
& $0.939\pm0.003$ & $4.142\pm0.410$ & $4.398\pm0.060$ \\
\bottomrule
\end{tabular}
\end{table*}

\begin{figure*}[t]
\centering
\includegraphics[width=0.95\linewidth]{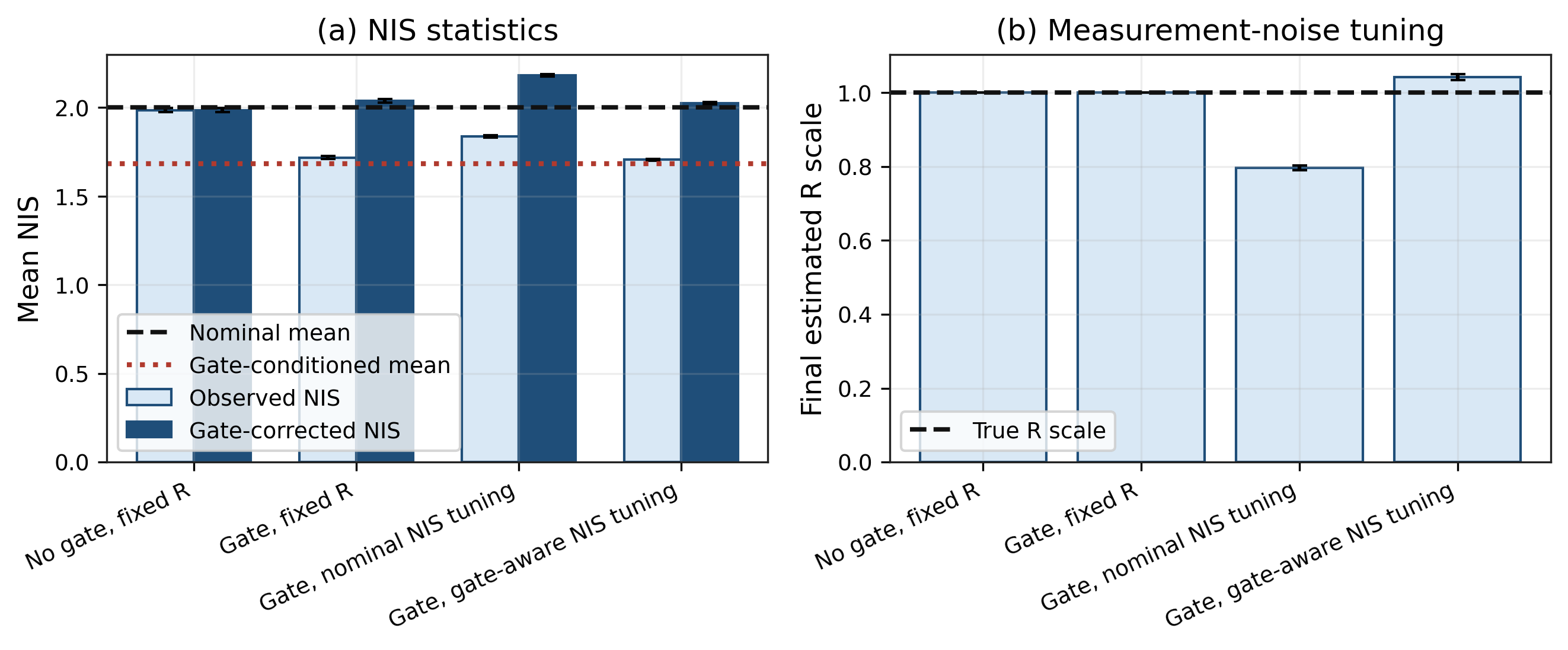}
\caption{\protect\HL{End-to-end tracking results for NIS statistics and adaptive
measurement-noise tuning. Nominal post-gate tuning learns a biased final
measurement-noise scale, whereas the gate-aware statistic remains close to the
true scale.}}
\label{fig:tracking_nis_r}
\end{figure*}

\begin{figure}[t]
\centering
\includegraphics[width=0.95\linewidth]{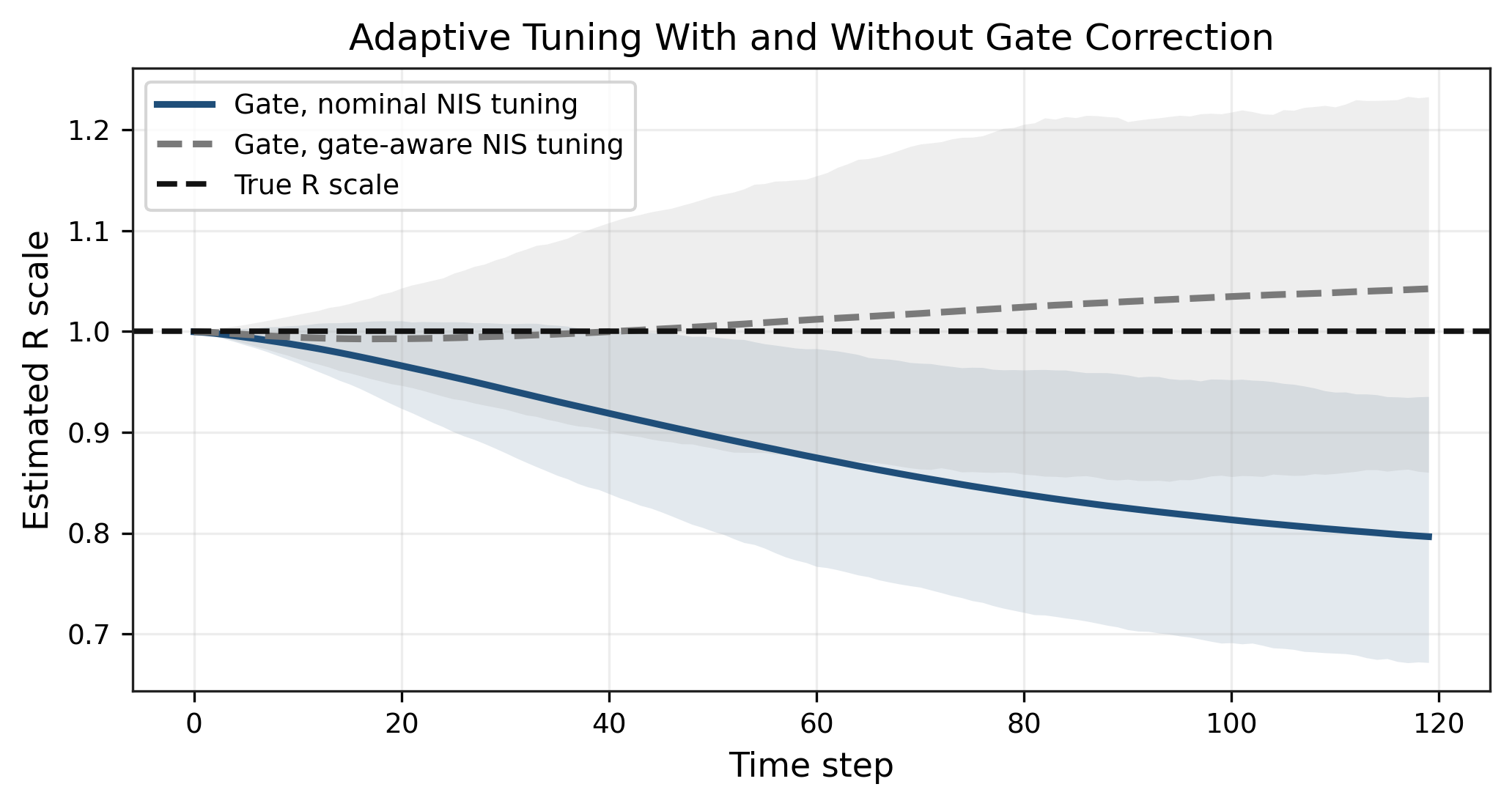}
\caption{\protect\HL{Time evolution of the estimated measurement-noise scale. Nominal
NIS tuning reacts to gate-contracted innovations by reducing the estimated
scale, while gate-aware tuning remains near the true reference.}}
\label{fig:r_tuning_time}
\end{figure}

\begin{figure*}[t]
\centering
\includegraphics[width=0.95\linewidth]{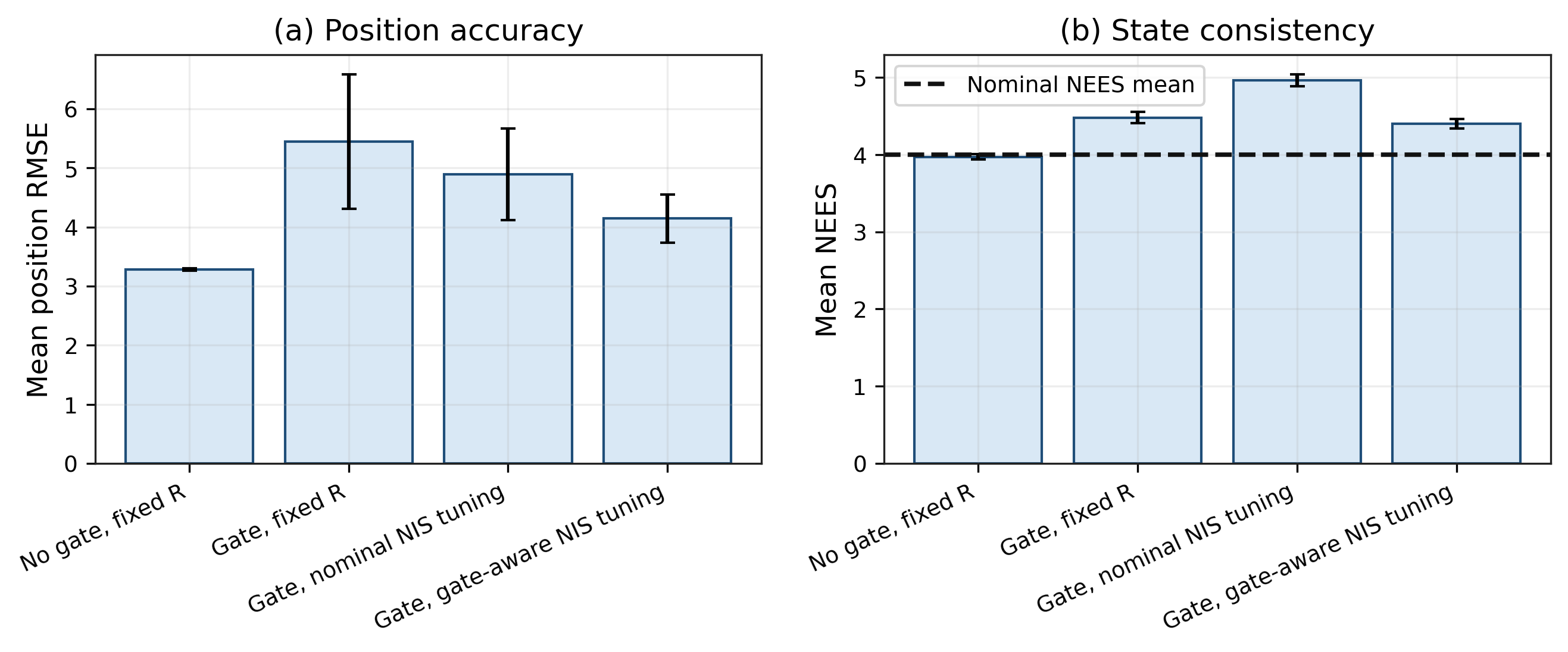}
\caption{\protect\HL{Position RMSE and NEES in the end-to-end tracking
experiment. The ungated case is a baseline that uses every measurement; it is
not the target that gated tuning is expected to outperform. The relevant tuning
comparison is between the two gated adaptive cases, where the gate-aware
reference reduces the bias induced by nominal post-gate NIS tuning. The NEES
panel also illustrates that NIS correction does not by itself guarantee exact
state-space consistency.}}
\label{fig:rmse_nees}
\end{figure*}

\HL{The end-to-end study confirms that the analytic gate-conditioned correction
has practical consequences for tuning.}
\HL{The ungated fixed-covariance filter has the lowest RMSE because it uses all
measurements and is included only as an ideal baseline without validation
rejection.}
\HL{The fixed gated filter has larger RMSE because valid measurements are
occasionally rejected by the gate, which is the expected cost of applying
validation.}
\HL{With nominal NIS tuning, the final measurement-noise scale is biased
downward to approximately 0.797, because the tuning rule attempts to force a
gate-contracted statistic toward an unconditional reference.}
\HL{With the gate-aware statistic, the final scale remains close to the true
value, approximately 1.042.}
\HL{Therefore, the performance-relevant comparison is within the gated adaptive
cases: gate-aware tuning yields lower mean RMSE and lower NEES than nominal
post-gate NIS tuning in this experiment.}
\HL{The NEES remains above the nominal value, emphasizing that NIS consistency
and NEES consistency are distinct properties.}

\subsection{NEES and Nonlinear Filters}

\HL{NIS consistency does not imply NEES consistency.}
\HL{Gating conditions not only the innovation statistic but also the event under
which a measurement update occurs.}
\HL{Consequently, the posterior state error after a gated update is also
conditioned on a selection event, and the resulting NEES need not follow the
unconditional chi-square reference even when the innovation correction restores
the mean accepted NIS.}
The Monte Carlo NEES results in Table~\ref{tab:tracking_results} and
Fig.~\ref{fig:rmse_nees}
\HL{illustrate this distinction.}

\HL{For nonlinear filters such as EKFs and UKFs, the exact linear-Gaussian
derivation does not apply globally, but the same selection mechanism operates
after local linearization or sigma-point prediction whenever an ellipsoidal
innovation gate is used.}
\HL{Problems such as bearings-only tracking can already exhibit fragile
chi-square approximations due to geometry and linearization error; validation
gating then further conditions the observed innovation sequence.}
\HL{A complete nonlinear theory is outside the scope of this paper, but the
linear result provides the baseline selection effect that any nonlinear
extension must account for.}

\section{Discussion}

The analysis developed in this paper provides a principled explanation for a
widely observed phenomenon in Kalman-based tracking systems: the tendency of
innovation-based diagnostics evaluated after gating and association to exhibit
systematically reduced innovation energy relative to nominal reference values,
even in well-functioning trackers.

Sections~III--VI show that this behavior arises from two structural selection
mechanisms inherent to practical tracking pipelines:
\begin{itemize}
\item ellipsoidal gating, which deterministically truncates the innovation
distribution and contracts its covariance, and
\item NN association, which introduces an additional order-statistic bias by
selecting the minimum-norm innovation among multiple candidates.
\end{itemize}

Both mechanisms operate even under ideal modeling assumptions.
They do not rely on clutter, false alarms, nonlinear dynamics, or model
mismatch.
As a result, tuning-based explanations alone are insufficient to account for
systematic bias observed in post-gate innovation statistics.

\HL{The Monte Carlo results show that the effect is not merely an asymptotic or
algebraic observation.}
\HL{In a multi-time-step tracking loop, using the unconditional NIS reference on
accepted innovations biases adaptive measurement-noise tuning, whereas using the
gate-conditioned correction substantially reduces that bias.}

\subsection{Implications for Classical Tracking Theory}

The apparent discrepancy between nominal and observed innovation statistics
arises because classical theory characterizes the unconditional innovation
process, whereas practical tracking systems operate on a conditionally observed
innovation stream.
The present work does not propose an alternative to adaptive or robust Kalman
filtering.
Rather, it provides a missing statistical layer that clarifies how
innovation-based diagnostics should be interpreted before invoking adaptation
mechanisms.
In this sense, the analysis complements existing tuning and robustness methods
by separating structural selection effects from genuine model mismatch.

\HL{The connection to the matrix-valued measures of Forsling et al.}
\cite{forsling2025matrix}
\HL{suggests an additional forward path.}
\HL{Scalar NIS gating compresses vector-valued innovation information into a
single quadratic form; matrix-valued consistency measures may retain richer
directional information and could be studied under analogous gate-conditioned
selection laws.}

\subsection{Implications for Implementation and Evaluation}

From an implementation perspective, the results suggest several practical
guidelines:
\begin{itemize}
\item Innovation covariance and NIS statistics computed after gating should be
interpreted using gate-conditioned reference values.
\item Systematically reduced innovation energy observed in post-gate diagnostics
does not necessarily indicate filter mistuning.
\item Adaptive tuning schemes based on innovation statistics should explicitly
account for selection-induced bias to avoid self-reinforcing measurement-noise
underestimation and valid-measurement rejection.
\item \HL{When ground truth is available, NEES should be evaluated separately
from NIS because correcting the accepted NIS reference does not guarantee
state-space consistency.}
\end{itemize}

These guidelines can be incorporated into existing tracking systems without
modifying the Kalman recursion, validation logic, or association rules
themselves.

\subsection{Limitations and Scope}

The analysis in this paper is deliberately restricted to linear measurement
models, Gaussian noise, ellipsoidal validation gates, and NN association.
These restrictions ensure that the identified effects are not artifacts of
nonlinearity or non-Gaussianity.

\HL{The Monte Carlo tracking experiment is also intentionally restricted to a
single-target, correctly specified, linear-Gaussian setting so that the observed
bias can be attributed to selection rather than to clutter modeling or
association failures.}
\HL{Extensions to nonlinear filters, PDA, MHT, and explicit clutter models are
important directions for future work, but require additional modeling choices
that would obscure the isolated selection mechanism studied here.}

\section{Conclusion}

Validation gating and NN association introduce predictable and quantifiable
biases in innovation statistics, even under ideal Kalman filter assumptions.
Ellipsoidal gating deterministically contracts the innovation covariance, while
NN association introduces an additional multiplicity-dependent bias through
order-statistic selection.

These structural effects explain why innovation-based diagnostics evaluated
after gating and association systematically deviate from nominal chi-square
references in practical tracking systems.
\HL{Monte Carlo validation confirms the analytic gate-conditioned and
order-statistic predictions, and an end-to-end tracking example demonstrates
that nominal NIS-based tuning can learn a biased measurement-noise scale when
applied to accepted post-gate innovations.}
\HL{Gate-aware interpretation restores the appropriate innovation reference,
although state-space consistency as measured by NEES remains a separate issue.}

\section*{Acknowledgment}
The author gratefully acknowledges Prof.\ Yaakov Bar-Shalom for his insightful
comments that helped clarify the interpretation and terminology of the paper.

\appendix

\section{Proof of Proposition~\ref{prop:gate_conditioned_moments}}

We prove the statements in Proposition~\ref{prop:gate_conditioned_moments} by
transforming the innovation into whitened coordinates and exploiting rotational
invariance of the Gaussian distribution.

Since $S \in \mathbb{S}_{++}^m$, there exists a symmetric matrix $S^{1/2}$ such
that $S^{1/2} S^{1/2} = S$.
Define the whitened innovation
\begin{equation}
u \triangleq S^{-1/2}\nu.
\label{eq:whitened_innovation}
\end{equation}
By construction,
\begin{equation}
u \sim \mathcal{N}(0,I_m),
\label{eq:u_distribution}
\end{equation}
and the NIS becomes
\begin{equation}
Z = \nu^\top S^{-1}\nu = u^\top u = \|u\|^2.
\label{eq:nis_whitened}
\end{equation}
Accordingly, the gating event can be written as
\begin{equation}
\mathcal{A} = \{ \|u\|^2 \le \tau \}.
\label{eq:gate_whitened}
\end{equation}

The distribution of $u$ is symmetric about the origin, and the event
$\|u\|^2 \le \tau$ is invariant under $u \mapsto -u$.
Therefore,
\[
\mathbb{E}[u \mid \mathcal{A}] = 0,
\]
and transforming back gives $\mathbb{E}[\nu \mid \mathcal{A}]=0$.

Define
\[
C \triangleq \mathbb{E}[u u^\top \mid \mathcal{A}].
\]
Since $u$ is rotationally invariant and the event $\|u\|^2 \le \tau$ depends
only on the norm of $u$, the conditional distribution $u\mid\mathcal{A}$ is also
rotationally invariant.
Consequently, $C$ commutes with all orthogonal transformations and must be of
the form $C=\gamma(\tau,m)I_m$.
Taking traces gives
\[
m\gamma(\tau,m)
=
\mathbb{E}[\tr(uu^\top)\mid\mathcal{A}]
=
\mathbb{E}[Z\mid Z\le\tau],
\]
which establishes \eqref{eq:gamma_definition}.
Transforming back to the original coordinates yields
\[
\mathbb{E}[\nu\nu^\top\mid\mathcal{A}]
=
S^{1/2} C S^{1/2}
=
\gamma(\tau,m)S.
\]
Finally, because $Z$ is nonnegative, nondegenerate, and has mean $m$, truncating
to the nontrivial event $Z\le\tau$ gives
$0<\mathbb{E}[Z\mid Z\le\tau]<m$, hence $0<\gamma(\tau,m)<1$.

\section{Proof of Proposition~\ref{prop:nn_energy_contraction}}

Condition on the gating acceptance event $\mathcal{A}$.
Let $\nu^{(1)},\dots,\nu^{(M)}$ be independent samples from
$(\nu \mid \mathcal{A})$, and define
\[
X_i \triangleq \|\nu^{(i)}\|^2,\qquad i=1,\dots,M.
\]
Then the $X_i$ are i.i.d. nonnegative random variables under the conditional
law.
NN association selects
\[
X_{(1)} \triangleq \|\nu^{(i^\star)}\|^2 = \min_{1\le i\le M}X_i.
\]
Since $X_{(1)}\le X_1$ almost surely,
\[
\mathbb{E}[X_{(1)}\mid\mathcal{A}]
\le
\mathbb{E}[X_1\mid\mathcal{A}]
=
\mathbb{E}[\|\nu\|^2\mid\mathcal{A}].
\]
If $M>1$ and the conditional distribution is nondegenerate, then
$\Pr(X_{(1)}<X_1\mid\mathcal{A})>0$, which yields strict inequality.

\bibliographystyle{IEEEtran}
\bibliography{references}

\begin{IEEEbiography}[{\includegraphics[width=1in,height=1.25in,clip,keepaspectratio]{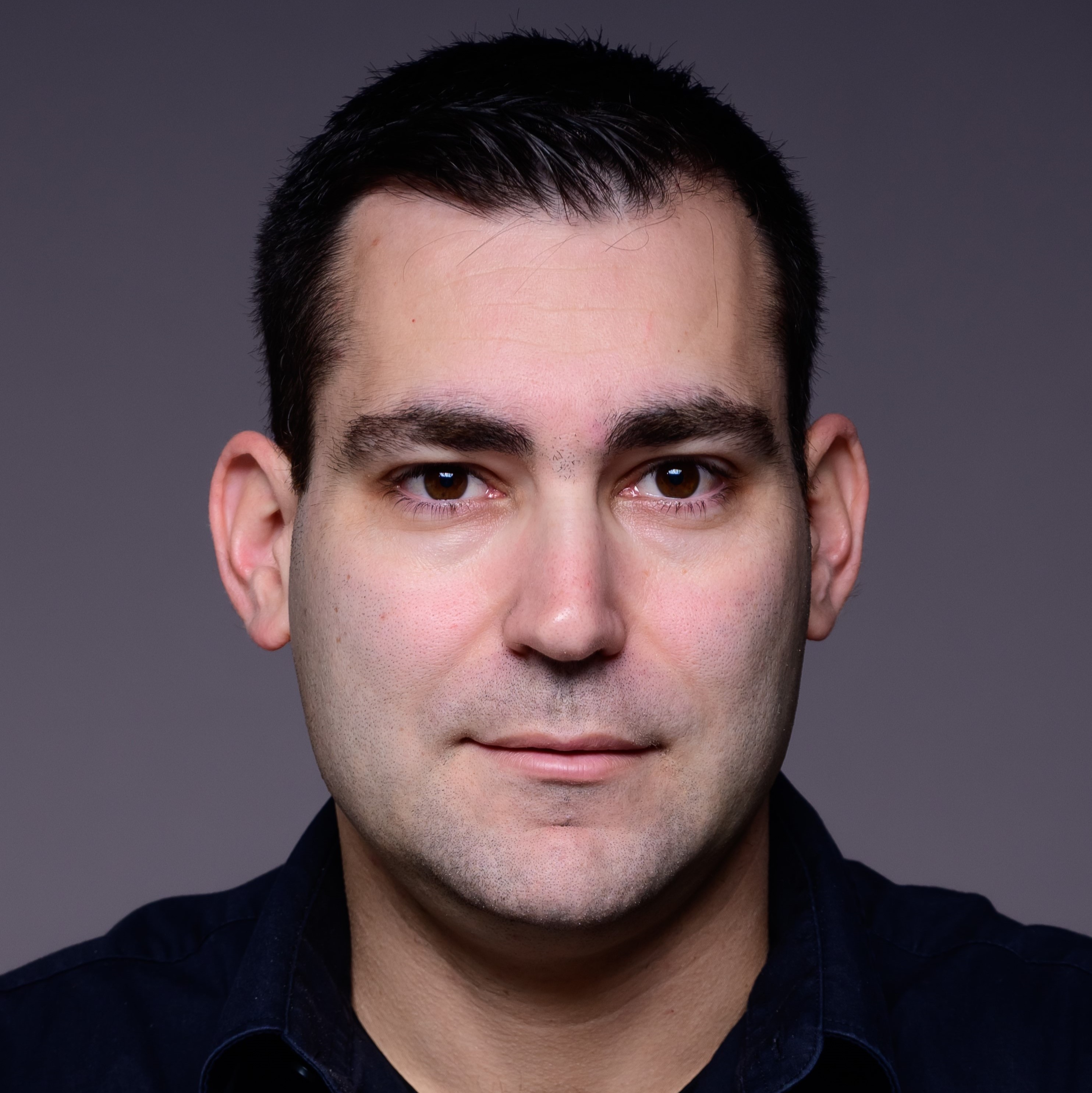}}]{Barak Or} (Member, IEEE) received a B.Sc. degree in aerospace engineering (2016), a B.A. degree (cum laude) in economics and management (2016), and an M.Sc. degree in aerospace engineering (2018) from the Technion–Israel Institute of Technology. He graduated with a Ph.D. degree from the University of Haifa, Haifa (2022).
His research interests include navigation, deep learning, sensor fusion, and estimation theory.
\end{IEEEbiography}

\end{document}